%% file: se2_control.tex
\documentclass{article}
\usepackage{arxiv}

\usepackage{amsthm}
\input{preamble_compatible}

\input{notation_defs}
\usepackage{graphicx}      
\usepackage{natbib}        
\usepackage{hyperref}

\newcommand{\publicationdetails}
{Under review}
\newcommand{\publicationversion}
{Preprint}
\begin{document}

\title{Exploiting spatial group error and synchrony for a unicycle tracking controller}
\headertitle{Exploiting spatial group error and synchrony for a unicycle tracking controller}

\author{
\href{https://orcid.org/0000-0001-2345-6789}{\includegraphics[scale=0.06]{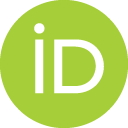}\hspace{1mm}
Matthew Hampsey}
\\
	Systems Theory and Robotics Group \\
	Australian National University \\
    ACT, 2601, Australia \\
	\texttt{matthew.hampsey@anu.edu.au} \\
	\And	\href{https://orcid.org/0000-0001-2345-6789}{\includegraphics[scale=0.06]{orcid.png}\hspace{1mm}
Pieter van Goor}
\\
	Systems Theory and Robotics Group \\
	Australian National University \\
    ACT, 2601, Australia \\
	\texttt{pieter.vangoor@anu.edu.au} \\
\And	\href{https://orcid.org/0000-0001-2345-6789}{\includegraphics[scale=0.06]{orcid.png}\hspace{1mm}
Robert Mahony}
\\
	Systems Theory and Robotics Group \\
	Australian National University \\
	ACT, 2601, Australia \\
	\texttt{robert.mahony@anu.edu.au} \\
}

\maketitle

\thispagestyle{empty}
\pagestyle{empty}

\begin{abstract}
Trajectory tracking for the kinematic unicycle has been heavily studied for several decades.
The unicycle admits a natural $\SE(2)$ symmetry, a key structure exploited in many of the most successful nonlinear controllers in the literature.
To the author's knowledge however, all prior work has used a body-fixed, or left-invariant, group error formulation for the study of the tracking problem.
In this paper, we consider the spatial, or right-invariant, group error in the design of a tracking controller for the kinematic unicycle.
We provide a physical interpretation of the right-invariant error and go on to show that the associated error dynamics are drift-free, a property that is not true for the body-fixed error.
We exploit this property to propose a simple nonlinear control scheme for the kinematic unicycle and prove almost-global asymptotic stability of this control scheme for a class of persistently exciting trajectories.
We also verify performance of this control scheme in simulation for an example trajectory.
\end{abstract}
\section{Introduction} \label{sec:Introduction}

The wheel is the mechanically simplest and most energy-efficient means of transporting mass on land.
Wheeled robots form an important subset of mobile robots, with applications across transportation, warehouse logistics, floor cleaning and hospitals (\cite{padenSurveyMotionPlanning2016}, \cite{kosticCollisionfreeTrackingControl2009}, \cite{asafaDevelopmentVacuumCleaner2018}, \cite{takahashiDevelopingMobileRobot2010}).
A task description for such applications chains together algorithms for planning, state estimation and control subtasks.
The control subtask is posed as a trajectory tracking problem, whereby a desired feasible trajectory is provided by the planning module and the control objective is to steer the vehicle to follow this trajectory (\cite{samsonModelingControlWheeled2008}).
The precise means by which is this is achieved largely depends on the wheel configuration and the actuation capabilities of the specific robot (\cite{samsonModelingControlWheeled2008}).

The kinematic model of the unicycle has been heavily studied in the systems and control literature (\cite{micaelliTrajectoryTrackingUnicycletype1993}).
The simplicity of the model coupled with the fact that it captures many of the key features associated with more complex systems make it an ideal test case for the development of sophisticated nonlinear control designs.
The model is nonholonomic, so by Brockett's criteria (\cite{brockettAsymptoticStabilityFeedback1983}) the set of stabilisable trajectories is limited by the analytic properties (such as smoothness) of any candidate feedback function.
Furthermore, the unicycle possesses a natural $\SE(2)$ symmetry, and this structure can be exploited in developing group error expressions on $\SE(2)$ that are core to many of the most popular control designs in the literature.
For example, in (\cite{kanayamaStableTrackingControl1990}), the $\SE(2)$ symmetry was used to define a left-invariant error and used as the basis for a nonlinear tracking control that stabilised the unicycle to certain classes of trajectories almost-globally.
In (\cite{micaelliTrajectoryTrackingUnicycletype1993}), a feedback-linearisation approach is used to propose a trajectory tracking control scheme for the unicycle.
In (\cite{panteleyExponentialTrackingControl1998}), an exponential tracking controller was designed by cascading linear controllers.
This controller required persistent excitation of the input angular rate.
In (\cite{jiangTrackingControlMobile1997}), an integrator backstepping approach is taken in order to design a globally stable controller.
More recently, in (\cite{rodriguez-cortesNewGeometricTrajectory2022}), a time-varying cascaded controller is used to provide almost-global tracking stability.
In (\cite{morinPracticalStabilizationDriftless2003}), a control scheme for the practical stabilisation for trajectories for general controllable driftless systems on Lie groups is investigated, including the example of the kinematic unicycle.

In this paper, we investigate the use of the spatial, or right-invariant, error for control design.
This perspective is a key contribution of the paper since to  the authors knowledge all the published control algorithms (for example, \cite{kanayamaStableTrackingControl1990}, \cite{micaelliTrajectoryTrackingUnicycletype1993}, \cite{panteleyExponentialTrackingControl1998}, \cite{jiangTrackingControlMobile1997}, \cite{morinPracticalStabilizationDriftless2003}, \cite{noijen*ObservercontrollerCombinationUnicycle2005}, \cite{meraSlidingmodeBasedController2020}, \cite{rochelTrajectoryTrackingUncertain2022}) use a body-fixed, or left-invariant, $\SE(2)$ error.
We provide physical insight into the difference between spatial and body group errors and go on to show that they lead to different error dynamics.
In particular, applying the desired velocity as a feed-forward input leads to synchronous error dynamics for the spatial group error.
This is in contrast to the body-fixed error dynamics that are time-varying, and may diverge, even if the vehicle is fed the correct desired input.
Synchrony of the error dynamics makes the design of the correction term straightforward since there is no need to dominate potential unstable error dynamics as is the case when the body group error is used.
We propose a simple Lyapunov function, and show that projecting the gradient of the Lyapunov function onto the actuated directions naturally leads to a gradient-based control design.
We prove almost-global asymptotic and local exponential stability of the control design for a class of persistently exciting bounded trajectories and provide numerical simulations that demonstrate the controller stabilises an example trajectory empirically.
The proposed method will generalise to all systems defined by left-invariant vector fields on general Lie groups, that is, all kinematic systems with Lie-group symmetries equipped with body-fixed actuators including marine vehicles, aerial vehicles, satellites, etc.
\section{Preliminaries} \label{sec:Preliminaries}
The matrix Lie group $\SO(2)$ is defined by the set of matrices
\begin{align*}
        \SO(2) = \left\{ R(\theta) \coloneqq \begin{pmatrix}
                \cos(\theta) & -\sin(\theta)\\
                \sin(\theta) & \cos(\theta)
        \end{pmatrix} : \theta \in [-\pi,\pi) \right\}.
\end{align*}
The associated Lie algebra $\so(2)$ is defined by the set
\begin{align*}
        \so(2) = \left\{ \Omega^\times \coloneqq \begin{pmatrix}
                0 & -\Omega\\
                \Omega & 0
        \end{pmatrix} : \Omega \in \R \right\},
\end{align*}
with the special element
\begin{align*}
        \mathbf{1}^\times \coloneqq \begin{pmatrix}
        0 & -1 \\ 1 & 0
\end{pmatrix}.
\end{align*}
The matrix Lie group $\SE(2)$ is defined by the set of matrices

\begin{align*}
        \SE(2) = \left\{ \begin{pmatrix}
                R & p\\
                0 & 1
        \end{pmatrix} : R \in \SO(2), p \in \R^2 \right\}.
\end{align*}

The associated Lie algebra $\se(2)$ is defined by the set
\begin{align*}
        \se(2) = \left\{ \begin{pmatrix}
                \Omega^\times & v\\
                0 & 0
        \end{pmatrix} : \Omega^\times \in \so(2), v \in \R^2 \right\}.
\end{align*}


There is a natural vector space isomorphism between $\R^3$ and $\se(2)$: define the mapping $(\cdot)^\wedge : \R^3 \to \se(2)$ by
\begin{align*}
        \begin{pmatrix} \Omega \\ v_x \\ v_y \end{pmatrix}^\wedge = \begin{pmatrix} 0 & -\Omega  & v_x \\ \Omega & 0 &v_y \\ 0 & 0 & 0 \end{pmatrix}.
\end{align*}
Let $(\cdot)^\vee : \se(2) \to \R^3$ be its inverse.
Similarly, given a linear map $A: \se(2) \to \se(2)$, let $A^\vee: \R^3 \to \R^3$ be the corresponding linear map on $\R^3$ defined by $A^\vee x  = (A (x^\wedge))^\vee$.

The Frobenius inner product is an inner product on $\R^{m \times m}$, defined by
\begin{align*}
        \langle A, B \rangle_{\mathbb{F}} = \tr(A^\top B).
\end{align*}
Note that $\langle A, BC \rangle_{\mathbb{F}} = \langle B^\top A, C \rangle_{\mathbb{F}}$.
Direct computation shows that
\begin{align*}
        \left\langle \begin{pmatrix}
        0 & -a & b \\ a & 0 & c\\ 0 & 0 & 0
\end{pmatrix}, \begin{pmatrix}
        0 & -d & e \\ d & 0 & f\\ 0 & 0 & 0
\end{pmatrix} \right\rangle_{\mathbb{F}} = 2ad + be + cf,
\end{align*}
so $\langle x^\wedge, y^\wedge \rangle_\mathbb{F} = \langle S x, y \rangle$, where $S = \diag(2, 1, 1)$.
The $\se(2)$ projection operator $\mathbb{P}_{\se(2)}: \R^{3 \times 3} \to \se(2)$ is given by
\begin{align*}
        \mathbb{P}_{\se(2)} \left( \begin{pmatrix}
                A_{2 \times 2} & x_{2 \times 1} \\ a & r
        \end{pmatrix} \right) = \begin{pmatrix}
                \frac{A_{2 \times 2} - A_{2 \times 2}^\top}{2} & x_{2 \times 1} \\ 0 & 0
        \end{pmatrix}.
\end{align*}

A function $F : \R_{\geq 0} \to \R^{n \times m}$ is called \textit{persistently exciting} if there exist real numbers $\varepsilon > 0$ and  $T > 0$ such that
\begin{align}
        \int_t^{t+T} F(\tau)^\top F(\tau) \td \tau \geq \varepsilon \Id_m,
\label{eq:persistent_excitation}
\end{align}
for all $t \in \R_{\geq 0}$.

\section{Problem Description} \label{sec:ProblemDescription}

\begin{figure}[!tb]
        \includegraphics[width=0.7\linewidth]{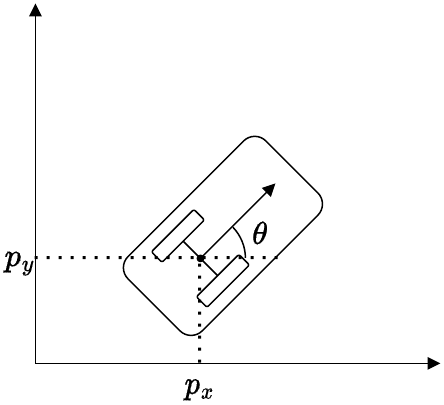}
        \centering
        \caption{Unicycle model}
        \label{fig:unicycle}
\end{figure}

Consider the standard kinematic unicycle (Fig. \ref{fig:unicycle}) with position $p = (p_x,p_y) \in \R^2$ , heading angle $\theta \in \Sph^1 \simeq [-\pi,\pi)$, and forward and angular inputs $v, \Omega \in \R$, respectively.
The system states evolve according to the kinematics
\begin{subequations}\label{eq:unicycle_kinematics}
\begin{align}
        \dot{p}_x &= v \cos (\theta)\\
        \dot{p}_y &= v \sin (\theta)\\
        \dot{\theta} &= \Omega.
\end{align}
\end{subequations}
The system state can be represented in the matrix Lie group $\SE(2)$ by
\begin{align} \label{eq:lie_group_state}
        X &= \begin{pmatrix}
                R(\theta) & p \\
                0_{1\times 2} & 1
        \end{pmatrix}
\end{align}
Using this representation, the system dynamics may be expressed as left-invariant dynamics on the group,
\begin{align}\label{eq:lie_group_dynamics}
        \dot{X} &= XU, &
        U &= \begin{pmatrix}
                \Omega \mathbf{1}^\times & v \eb_1 \\
                0_{1\times 2} & 0
        \end{pmatrix} \in \se(2).
\end{align}
Additionally, by defining
\begin{align}
        B = \begin{pmatrix}
                1 & 0 \\ 0 & 1 \\ 0 & 0
        \end{pmatrix} \text{ and }
        u = \begin{pmatrix}
                \Omega \\ v
        \end{pmatrix} \label{eq:little_u},
\end{align}
one has $U = (Bu)^\wedge$ and
\begin{align*}
        \dot{X} = X(Bu)^\wedge.
\end{align*}

In this paper, we address the problem of tracking a desired trajectory of the unicycle.
Let $\theta_d \in \Sph^1$ and $p_d \coloneqq ( p_x^d,  p_y^d) \in \R^2$ denote the desired heading angle and position of the unicycle, and let $X_d \in \SE(2)$ denote the representation of the desired state in the Lie group as in \eqref{eq:lie_group_state}.
Likewise, let $\Omega_d, v_d \in \R$ denote the inputs associated with the desired trajectory, let $U_d \in \se(2)$ denote their representation in the Lie algebra as in \eqref{eq:lie_group_dynamics} and let $u_d$ denote their representation in $\R^2$ as in \eqref{eq:little_u}.
Then the desired trajectory dynamics are
\begin{align}\label{eq:desired_trajectory_dynamics}
        \dot{X}_d = X_d U_d = X_d (B u_d)^\wedge.
\end{align}
The problem becomes that of finding an admissible control $U \in \se(2)$ so that the true system state $X$ converges to and tracks the desired system state $X_d$.

\section{Lie Group Errors} \label{sec:LieGroupErrors}

The most common design methodology for constructive nonlinear control does not directly try to control $X(t) \to X_d(t)$.
Rather the approach taken is to define an error $E(X(t), X_d(t))$ and study the problem of driving $E(X(t), X_d(t)) \to E_\star$ where $E_\star$ is some constant reference such that when $E = E_\star$ then $X = X_d$.
On any Lie group $\grpG$ there are two natural Lie group errors that can be used for this role.

\noindent{\textbf{Body-fixed group error:}}
The body fixed (or left-invariant) group error
\begin{align}
E_L & \coloneqq X_d^{-1} X
 = \begin{pmatrix}
R_d^\top R & R_d^\top (p - p_d)\\
0 & 1
\end{pmatrix}. \label{eq:EL_components}
\end{align}
This error can be interpreted as the $\SE(2)$ transformation taking the frame $X_d$ to $X$, expressed in coordinates of the frame $X_d$ (Figure \ref{fig:leftInvariantError}).

\noindent{\textbf{Spatial group error:}}
The spatial (or right-invariant) group error
\begin{align}
E_R &\coloneqq X X_d^{-1}
= \begin{pmatrix}
R R_d^\top & p - R R_d^\top p_d \\ 0_{1\times 2} & 1
\end{pmatrix}. \label{eq:error_definition}
\end{align}
This error can be interpreted as the $\SE(2)$ transformation taking the frame $X_d$ to $X$, expressed in reference coordinates (Figure~\ref{fig:rightInvariantError}).

Note that in both cases the group error encodes the $\SE(2)$ transformation that moves the desired state to the robot state.
Clearly, if this transformation is identity ($E = I_3 \in \SE(2)$) then $X = X_d$ for either error definition.
However, the physical transformations are quite different as seen in \eqref{eq:EL_components} and \eqref{eq:error_definition}.
For the body-fixed transformation the rotation and translation are decoupled since the rotation is undertaken around the body reference.
Conversely, in the spatial transformation, the rotation is undertaken around the origin of the reference frame and moves the frame.

The body-fixed error has been the natural choice for tracking control design for several reasons.
It is the group error formulation taught in most text books and corresponds to the coordinate change formula that most roboticists use to understand rigid-body transformations.
It is also the natural error to encode a rigid-body transformation from the perspective of the robot itself.
In contrast, the spatial error representation is less commonly used in mainstream robotics, although it is core to the field of screws/twists and is used in exponential coordinates.
The spatial error also depends on the reference frame and as such is not intrinsic to the motion of the vehicles.
The justification for considering the spatial group error comes from studying the error dynamics.

\begin{lemma}
Let $X_d(t) \in \SE(2)$ be a trajectory satisfying \eqref{eq:desired_trajectory_dynamics} and let $X(t) \in \SE(2)$ be a trajectory satisfying \eqref{eq:lie_group_dynamics}.
The error dynamics of the body-fixed error \eqref{eq:EL_components} are given by
\begin{align*}
\dot{E}_L = -U_d E_L + E_LU.
\end{align*}
\end{lemma}

\begin{proof}
By straightforward computation,
\begin{align*}
\dot{E}_L &= -X_d^{-1}\dot{X}_d X_d^{-1}X + X_d^{-1}\dot{X}\\
                &= -X_d^{-1} X_d U_d X_d^{-1}X + X_d^{-1}XU\\
                &= -U_d E_L + E_L U.
\end{align*}
\end{proof}

\begin{figure}[!tb]
        \includegraphics[width=0.7\linewidth]{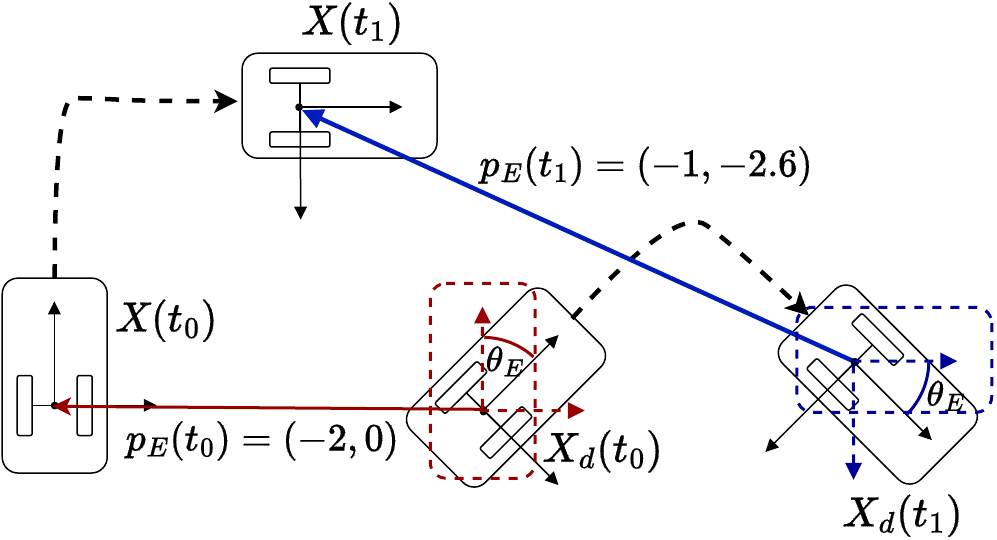}
        \centering
        \caption{Visualisation of the body-fixed, or left-invariant, error.
        The relative transformation between $X_d$ and $X$ in both cases involves an anti-clockwise rotation of 45 degrees. However, at $t_0$ the translation is (-2, 0) while at $t_1$ the translation is (-1, -2.6) demonstrating that the error is not preserved.}
        \label{fig:leftInvariantError}
\end{figure}

Consider applying the reference input $U = U_d$ as a feedforward compensation.
Then the body group error evolves according to  $\dot{E} = E_L U_d - U_d E_L$.
The evolution of the error term is visualised in Figure \ref{fig:leftInvariantError} and it is clear that the system is not synchronous, that is, the error is not preserved under feed-forward control.
Here $X_d(t_0)$ is rotated anti-clockwise 45 degrees and translated by (-2, 0) (in the frame of $X_d(t_0)$) to arrive at $X(t_0)$
\begin{align*}
X_d^{-1}(t_0) X(t_0) = \begin{pmatrix}
0.717 & - 0.717 & -2 \\ 0.717 & 0.717 & 0 \\ 0 & 0 & 1
\end{pmatrix}
\end{align*}
At time $t_1$, $X_d(t_1)$ is rotated anti-clockwise 45 degrees and then translated by (-1, -2.6) (in the frame of $X_d(t_1)$) to arrive at $X(t_1)$.
\begin{align*}
X_d^{-1}(t_1)X (t_1)  = \begin{pmatrix}
0.717 & - 0.717 & -1 \\ 0.717 & 0.717 & -2.6 \\ 0 & 0 & 1
\end{pmatrix}.
\end{align*}
In the context of control design, this non-zero term
$\dot{E} = E_L U_d - U_d E_L$ in the error dynamics presents as a drift term  that must be compensated by the control action.
If the system were fully-actuated, then the exogenous dynamics component $U_d E$ can be compensated by setting
\begin{align*}
U = \Ad_{E_L^{-1}}U_d + E_L^{-1} \Delta,
\end{align*}
where $\Delta$ is then the additional input that drives $E_L \to I_3$.
For example, see \cite{chaturvediRigidBodyAttitudeControl2011}, \cite{leeGeometricTrackingControl2010}.
However, if the system is under-actuated, as with the kinematic unicycle, the term $\Ad_{E_L^{-1}}U_d$ generally does not lie in the actuated directions.  
In this case, the exogenous dynamics cannot be directly compensated and the control must be used to dominate the effects of the drift in the Lyapunov analysis.
There are various well established control algorithms that take this approach in the literature, such as \cite{kanayamaStableTrackingControl1990} or \cite{leeTrackingControlUnicyclemodeled2001}.

On the other hand, consider the choice of the spatial or right-variant error $E_R \coloneqq X X_d^{-1}$.

\begin{lemma}
Let $X_d(t) \in \SE(2)$ be a trajectory satisfying \eqref{eq:desired_trajectory_dynamics} and let $X(t) \in \SE(2)$ be a trajectory satisfying \eqref{eq:lie_group_dynamics}.
Define the control difference $\tilde{U} \coloneqq U - U_d \in \se(2)$ (and $\tilde{u} \coloneqq u - u_d \in \R^2$, so that $(B\tilde{u})^\wedge = \tilde{U}$).
Then the dynamics of $E_R$ \eqref{eq:error_definition} are given by
\begin{align}\label{eq:error_dynamics}
        \dot{E}_R &= E_R \Ad_{X_d} \tilde{U} = E_R \Ad_{X_d} (B\tilde{u})^\wedge.
\end{align}
\end{lemma}

\begin{proof}
        The identity \eqref{eq:error_dynamics} follows from straight-forward computation.
        Taking the time derivative of \eqref{eq:error_definition}:
        \begin{align*}
                \dot{E}_R &= \dot{X}X_d^{-1} - XX_d^{-1}\dot{X}_dX_d^{-1}\\
                        &= XUX_d^{-1} - XX_d^{-1}X_dU_dX_d^{-1}\\
                        &= XX_d^{-1}X_d(U - U_d)X_d^{-1}\\
                        &= E_R \Ad_{X_d}(U - U_d)\\
                        &= E_R \Ad_{X_d}\tilde{U}\\
                        &= E_R \Ad_{X_d}(B\tilde{u})^\wedge,
        \end{align*}
        as required.

\end{proof}

\begin{figure}[!tb]
        \includegraphics[width=0.7\linewidth]{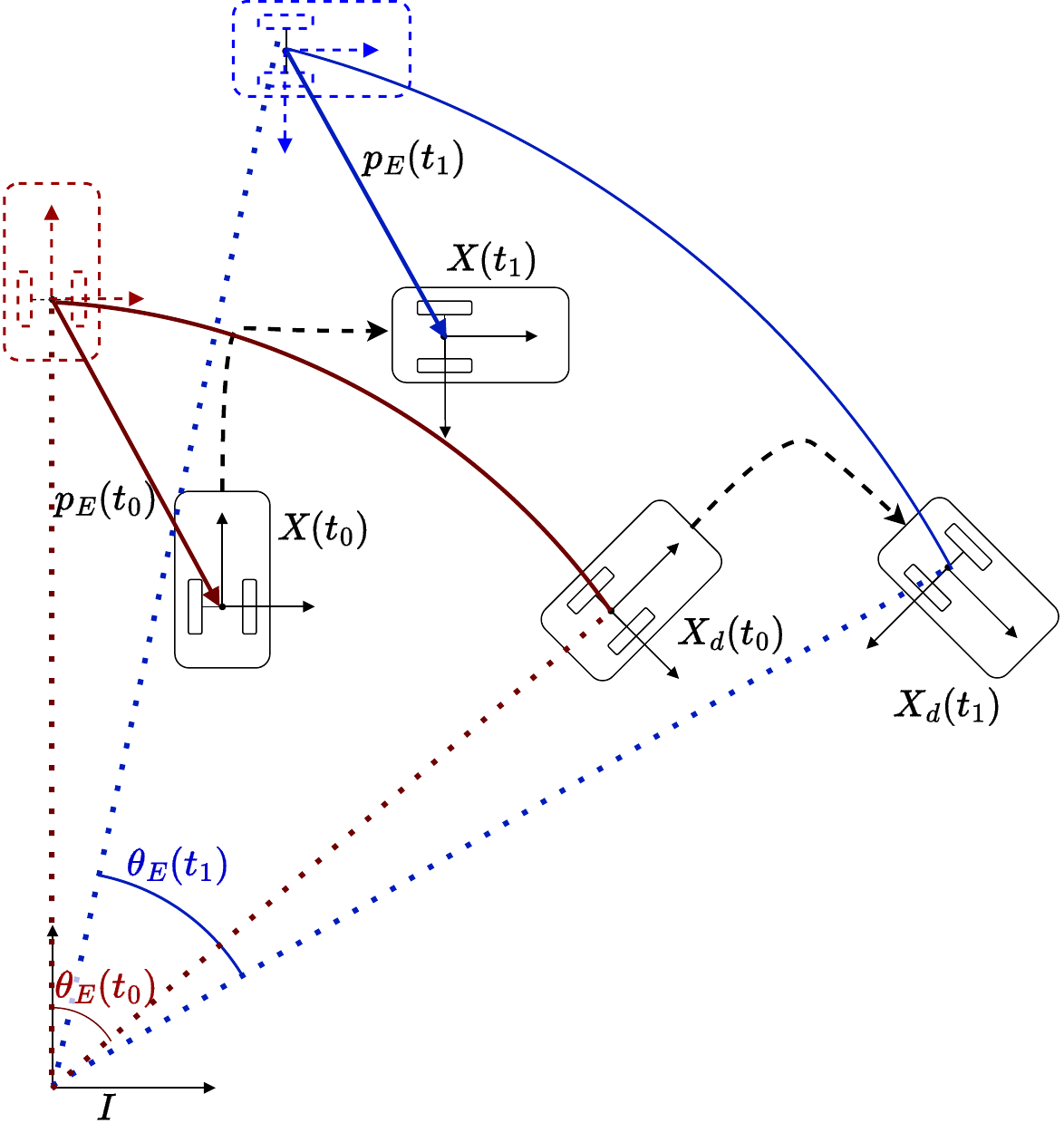}
        \centering
        \caption{Visualisation of the spatial (or right-invariant) error.  In both cases the $X(t_0)$ and $X(t_1)$ are rotated around the reference origin by 45 degrees and then translated (in the reference frame) in the $(1,-1)$ direction. }
        \label{fig:rightInvariantError}
\end{figure}

Setting $U = U_d$, then $\tilde{U} = 0$.
It follows that the feed-forward system is synchronous, that is, $\dot{E}_R = 0$.
This is visualised in  Figure \ref{fig:rightInvariantError}.
Note that since the error $\dot{E}_R = 0$ then $E_R(t_1) = E_R(t_0)$.
It is also clear from Figure \ref{fig:rightInvariantError} that constant spatial error does not correspond to a ``constant local distance'' between the desired trajectory and vehicle trajectory.
Studying the spatial error leads to simpler error dynamics but introduces a more complex interpretation of the meaning of the error, and in particular, introduces dependence on the reference frame choice.
Although $E_R = I_3$ guarantees that $X = X_d$ the dependence on $E_R$ on the reference frame means that analysing convergence is more complex.
The following Lemma is important in ensuring that convergence of $E_r(t) \to I_3$ ensures $X(t) \to X_d(t)$.

%
%
%

\begin{lemma}
        Let $X_d(t) \in \SE(2)$ be a trajectory satisfying \eqref{eq:desired_trajectory_dynamics} and let $X(t) \in \SE(2)$ be a trajectory satisfying \eqref{eq:lie_group_dynamics}.
        If $p_d $ is bounded then $E_R(t) \to \Id_3$ if and only if $X(t) \to X_d(t)$.
\end{lemma}

\begin{proof}
First, assume that $\lim_{t \to \infty} E_R(t) = \Id_3$.
Let $\varepsilon > 0$ be given.
We need to show that there exists a $T > 0$ such that $\lVert R(t) - R_d(t) \rVert < \varepsilon$ and $\lVert p - p_d \rVert < \varepsilon$, for $t > T$.
The set $\SO(2)$ is compact, so $R_d(t)$ is bounded.
Let $M = \sup_t \lVert R_d(t) \rVert$ and $N = \sup_t \lVert p_d(t) \rVert$.
Then, because $\lim_{t \to \infty} E_R(t) = \Id_3$, there exists a $T > 0$ such that $\lVert R(t) R_d(t)^\top - \Id_2 \rVert < \min\{\frac{\varepsilon}{M}, \frac{\varepsilon}{2N}\}$ and $\lVert p(t) - R_E(t) p_d(t) \rVert < \frac{\varepsilon}{2}$ for $t > T$.
Therefore, the following inequalities hold:
\begin{align*}
        \lVert R - R_d \rVert &= \lVert (R - R_d)R_d^T R_d \rVert \\
                                &\leq \lVert (R - R_d)R_d^T \rVert \lVert R_d \rVert\\
                                &= \lVert E - I \rVert \lVert R_d \rVert\\
                                &< \frac{\varepsilon}{M} M = \varepsilon,
\end{align*}
and
\begin{align*}
        \lVert p - p_d \rVert &\leq \lVert p - R_E p_d \rVert + \lVert R_E p_d - p_d \rVert \\
                                &\leq \lVert p - R_E p_d \rVert + \lVert R_E - I_2 \rVert \lVert p_d \rVert\\
                                &< \frac{\varepsilon}{2} + \frac{\varepsilon}{2N} N = \varepsilon,
\end{align*}
so $X \to X_d$.
The converse statement is proved by a similar argument.
\end{proof}

\section{Constructive Lyapunov control on $\SE(2)$} \label{sec:ConstructiveLyapunov}

The synchrony property of the spatial error is a key property exploited in the design of the tracking controller.
Specifically, we choose a candidate Lyapunov function $\calL (E)$.
Then the correction term is generated by projecting the gradient of $\calL (E)$ onto the actuation directions of the vehicle.
Even if this projection leads to a null correction term, the synchrony of the error ensures that the Lyapunov function never increases.
We then depend on excitation of the reference trajectory to provide global asymptotic stability of the error dynamics.

We approach the controller design constructively: that is, we define a candidate Lyapunov function, and use this to derive a control law.
\begin{theorem}
Let $X_d(t) \in \SE(2)$ denote the desired system trajectory with desired input $u_d \in \R^2$ as in \eqref{eq:desired_trajectory_dynamics}, and let $X(t) \in \SE(2)$ denote the true system state.
Let $E = X X_d^{-1}$ denote the error \eqref{eq:error_definition} and define 
\begin{align*}
        S &= \diag(2,1,1) \in \R^{3\times3}.
\end{align*}
Define the control input delta $\tilde{u} \in \R^2$ by
\begin{align}
        \tilde{u} &= -B^\top \Ad_{X_d}^{\vee \top} S \mathbb{P}_{\se(2)}(E^\top E - E^\top)^\vee,
\label{eq:u_tilde}
\end{align}
or in components, as
\begin{align*}
        \tilde{\Omega} &= -2 \sin(\theta_E) - p_d^\top \mathbf{1}^\times R_E^\top p_E\\
        \tilde{v} &= -e_1^\top R_d^\top R_E^\top p_E. 
\end{align*}
Then, the identity $I$ is a stable equilibrium point of the error system $\dot{E}$.
Moreover, if $B^\top \Ad_{X_d(t)}^\top S$ is persistently exciting \eqref{eq:persistent_excitation} then:
\begin{enumerate}
        \item The error $E$ converges to $I$ almost-globally asymptotically.  The complement of the basin of attraction is the singleton set $\{\calE_u\}$, where 
        \begin{align}
        \calE_u =  \begin{pmatrix}
                R(\pi) & 0 \\ 0 & 1 \end{pmatrix}.
        \end{align}
        \item The error $E$ converges to $I$ locally exponentially.
\end{enumerate}
\end{theorem}
\begin{proof}
We begin the proof of claim (1) by defining the candidate Lyapunov function
\begin{align}
\mathcal{L} (E) = \frac{1}{2} \langle E - I, E - I \rangle_{\mathbb{F}} 
= \tr\left( (E - I)^\top (E - I) \right),
\label{eq:lyapunov_function}
\end{align}
using the matrix or Frobenius inner product. 
Differentiating yields 
\begin{align*}
        \dot{\mathcal{L}} (t) = \langle E-I, \dot{E} \rangle_{\mathbb{F}} = \langle E-I, E\Ad_{X_d} (B\tilde{u})^\wedge \rangle_{\mathbb{F}}.
\end{align*}
Rearranging terms inside the Frobenius inner product yields 
\begin{align*}
\langle E-I, E\Ad_{X_d} (B\tilde{u})^\wedge \rangle_{\mathbb{F}} = \langle E^\top(E-I), \Ad_{X_d} (B\tilde{u})^\wedge \rangle_{\mathbb{F}}.
\end{align*}
For all $X_d \in \SE(2)$ and $\tilde{U} = (B\tilde{u})^\wedge  \in \se(2)$ then $\Ad_{X_d} \tilde{U} \in \se(2)$.
Since the right-hand side lies in  $\se(2)$ the inner product is unchanged if the left-hand side is projected onto $\se(2)$.
That is,
\begin{align*}
        \dot{\mathcal{L}} (t) = \langle \mathbb{P}_{\se(2)}(E^\top E-E^\top), \Ad_{X_d} (B\tilde{u})^\wedge \rangle_{\mathbb{F}}.
\end{align*}
Note that for the Frobenius inner product 
\[
\langle A, B \rangle_{\mathbb{F}} = \tr(A^\top B) 
=  (A^\vee)^\top S B^\vee 
= \langle SA^\vee, B^\vee \rangle 
\]
since the skew-symmetric components of $\mathbf{1}^\times$ get counted twice in the Frobenius inner product.  
Thus,
\begin{align*}
\dot{\cal{L}}(E) &= \langle \mathbb{P}_{\se(2)}(E^\top E-E^\top), \Ad_{X_d} (B\tilde{u})^\wedge \rangle_{\mathbb{F}} \\
&= \langle S \mathbb{P}_{\se(2)}(E^\top E-E^\top)^\vee, \Ad_{X_d}^{\vee} B\tilde{u} \rangle\\
&= \langle B^\top \Ad_{X_d}^{\vee \top} S \mathbb{P}_{\se(2)}(E^\top E-E^\top)^\vee, \tilde{u} \rangle.
\end{align*}
Substituting the control \eqref{eq:u_tilde} yields 
\begin{align*}
        \dot{\cal{L}}(E) = -\lVert B^\top \Ad_{X_d}^{\vee \top} S \mathbb{P}_{\se(2)}(E^\top E-E^\top)^\vee \rVert^2. 
\end{align*}
This shows that $\calL$ is non-increasing and proves the error dynamics are stable.

Since $\calL$ is non-increasing then $\calL (t) \leq \calL (0)$ for all $t$ and $\calL (t) \to c$ for some positive constant $c$.
Rewriting $\calL (t)$ as $\calL(t) = 2(1 - \cos \theta_E) + \frac{1}{2} \lVert p_E \rVert^2$ implies that $p_E$ is bounded.
For notational convenience, define
\begin{align*}
    A(t) &\coloneqq B^\top \Ad_{X_d(t)}^{\vee \top} S & \in \R^{2 \times 3}\\
    C(t) &\coloneqq \mathbb{P}_{\se(2)}(E(t)^\top E(t)-E(t)^\top)^\vee & \in \R^3 
\end{align*}
and note that $\dot{\calL} = -\lVert A(t) C(t)\rVert^2$.
Taking the second derivative of $\calL$ and using the Cauchy-Schwarz inequality and submultiplicativity of matrix norms one has
\begin{align*}
\lVert \ddot{\cal{L}} (t) \rVert 
&= 2  \langle  A(t)C(t),
\dot{A}(t)C(t) + A(t)\dot{C}(t)  \rangle \\
& \leq 2 \lVert A(t)C(t) \rVert \lVert \dot{A}(t)C(t) + A(t)\dot{C}(t) \rVert \\
& \leq 2 \lVert A(t) \rVert \lVert C(t) \rVert 
\left(\lVert \dot{A}(t) \rVert  
\lVert C(t) \rVert + \lVert A(t) \rVert  \lVert \dot{C}(t) \rVert \right).
\end{align*}
Thus, in order to show that $\ddot{\cal{L}} (t)$ is bounded, it suffices to show that each of $A(t), C(t), \dot{A}(t)$ and $\dot{C}(t)$ is bounded.
We have:
\begin{align*}
        A(t) &=  \begin{pmatrix} 2 & p_d(t)^\top \mathbf{1}^\times \\ 0 & e_1^\top R_d(t)^\top  \end{pmatrix},\\
        \dot{A}(t) &= \begin{pmatrix} 0 & v_d^\top R_d^\top \mathbf{1}^\times \\ 0 & -e_1^\top \Omega_d^\times R_d^\top  \end{pmatrix},   \\
        C(t) &= \begin{pmatrix} \sin \theta_E(t) \\ R_E(t)^\top p_E(t) \end{pmatrix},\\
        \dot{C}(t) &= \begin{pmatrix} \cos \theta_E(t) \tilde{\Omega} \\ -\tilde{\Omega}^\times \left(R_E(t)^\top p_E(t) + p_d(t) \right) + R_d(t) \tilde{v}(t) \end{pmatrix},\end{align*}
and so
\begin{align*}
        \lVert A(t) \rVert^2 &= 5 + \lVert p_d \rVert^2,\\
        \lVert \dot{A}(t) \rVert^2 &=  \lVert v_d \rVert^2  +  \lVert e_1^\top \Omega_d^\times \rVert^2\\
                                  &\leq \lVert v_d \rVert^2  + 2\Omega_d^2\\
        \lVert C(t) \rVert^2 &= \sin^2 \theta_E(t) + \lVert p_E \rVert^2,\\
        \lVert \dot{C}(t) \rVert^2 &= \cos^2 \theta_E(t) \tilde{\Omega}^2 + \lVert -\tilde{\Omega}^\times \left(R_E(t)^\top p_E(t) + p_d(t) \right) + R_d(t) \tilde{v}(t) \rVert^2\\
                         & \leq \tilde{\Omega} ^ 2 + 2 \tilde{\Omega}^2( \lVert p_E(t) \rVert + \lVert p_d(t) \rVert) + \lVert \tilde{v} \rVert^2.
\end{align*}
Each of $x_d, v_d, \Omega_d$ and $x_E$ is bounded, so $A(t)$, $\dot{A}(t)$ and $C(t)$ are bounded, and hence $\tilde{U} = A(t)C(t)$ is bounded.
Thus, $\tilde{\Omega}$ and $\tilde{v}$ are bounded and so $\dot{C}(t)$ is also bounded.
Thus, by Barbalat's lemma (\cite{slotine1991applied}), $\dot{\calL}(t) \to 0$.
This also implies that $\tilde{U} = A(t)C(t)$ approaches zero, and so $\dot{E} \to 0$ and $\dot{C}(t) = \frac{\td}{\td t}\mathbb{P}_{\se(2)}(E(t)^\top E(t)-E(t)^\top)^\vee \to 0$.


Since $A(t)$ is persistently exciting (Lemma \ref{lemma:persistent_excitation}), it now follows that $C(t) \coloneqq \mathbb{P}_{\se(2)}(E(t)^\top E(t)-E(t)^\top)^\vee \to 0$.
The map $E \mapsto \mathbb{P}_{\se(2)}(E^\top E-E^\top)^\vee$ can be expressed in components as $(\sin \theta_E, R_E^\top p_E)$, so $p_E \to 0$ for all $p_E$ and $\sin \theta_E \to 0$.
Hence $\theta_E$ asymptotically approaches either $0$ or $\pi$.

To see that the point $(\pi, 0)$ is an unstable equilibrium point, it suffices to show that every neighbourhood of $(\pi, 0)$ contains a point $q$ with $\calL (q) < \calL (\pi, 0)$.
Every neighborhood of $(\pi, 0)$ contains a point $(\pi - \varepsilon, 0)$ for small enough $\varepsilon > 0$.
In general, $\calL (\theta, p) = 2(1 - \cos(\theta)) + \frac{1}{2} \lVert p \rVert^2$, so $\calL (\pi, 0) = 4$.
At a point $(\pi - \varepsilon, 0)$, $\calL (\pi - \varepsilon, 0) = 2(1 - \cos(\pi - \varepsilon)) = 2(1 +\varepsilon) < 4$.
Thus, the equilibrium $(\pi, 0)$ is unstable.
This proves claim (1). 

To prove claim (2), define the local coordinates $\mu = (2\theta_E, p_E)$ and define
\begin{align*}
P(t) \coloneqq \begin{pmatrix}
        1 & p_d^\top \mathbf{1}^\times \\
        -\mathbf{1}^\times p_d & -\mathbf{1}^\times p_d p_d^\top \mathbf{1}^\times + R_d e_1 e_1^\top R_d^\top
        \end{pmatrix}.
\end{align*}
The linearisation of the error dynamics around $\mu = (0, 0)$ can be computed to be:
\begin{align}
\begin{pmatrix}
\dot{\mu}_\theta \\ \dot{\mu}_p
\end{pmatrix} &=
- P(t) \begin{pmatrix}
        \mu_\theta \\ \mu_p
        \end{pmatrix} \label{eq:lin_sys}\\
&= -\begin{pmatrix} 1 & 0 \\ 0 & 2 \Id \end{pmatrix} \Ad_{X_d}^\vee B B^\top \Ad_{X_d}^{\vee \top} \begin{pmatrix} \frac{1}{4} & 0 \\ 0 & \frac{1}{2} \Id \end{pmatrix} \begin{pmatrix}
        \mu_\theta \\ \mu_p
        \end{pmatrix} \notag \\
&= -4 \begin{pmatrix} \frac{1}{4} & 0 \\ 0 & \frac{1}{2} \Id \end{pmatrix} \Ad_{X_d}^\vee B B^\top \Ad_{X_d}^{\vee \top} \begin{pmatrix} \frac{1}{4} & 0 \\ 0 & \frac{1}{2} \Id \end{pmatrix} \begin{pmatrix}
        \mu_\theta \\ \mu_p
        \end{pmatrix} \notag.
\end{align}

By assumption, $B^\top \Ad_{X_d(t)}^{\vee \top}$ is persistently exciting, so there exists some lower bound $\varepsilon > 0$ such that
\begin{align*}
        \int_t^{t+T} \Ad_{X_d}^\vee B B^\top \Ad_{X_d}^{\vee \top} \td \tau \geq \varepsilon \Id.
\end{align*}
The smallest eigenvalue of the matrix $\begin{pmatrix} \frac{1}{4} & 0 \\ 0 & \frac{1}{2} \Id \end{pmatrix}\begin{pmatrix} \frac{1}{4} & 0 \\ 0 & \frac{1}{2} \Id \end{pmatrix}$ is $\frac{1}{16}$, so it also holds that
\begin{align*}
\int_t^{t+T} P(\tau) \td \tau &\geq \frac{1}{4} \varepsilon \Id.
\end{align*}
As the matrix $P(t)$ is also positive semi-definite and symmetric, by Lemma \ref{prop:excited_linear}, the linearised system \eqref{eq:lin_sys} is uniformly exponentially stable.
\end{proof}

\section{Simulation Experiments} \label{sec:Simulation}

In order to empirically verify that the controller works, the system is implemented in simulation for an elliptical trajectory.
In general, an elliptical trajectory has the form
\begin{align}
        p_d(t) = \begin{pmatrix}
        a \cos(ht) \\ b \sin(ht) \label{eq:ellipse}
\end{pmatrix},
\end{align}
implying that $\Omega_d(t) = h, \theta_d(t) = ht$.
\subsection{Excitation of elliptical trajectories}
For an elliptical trajectory of the form \eqref{eq:ellipse}, direct computation shows that
\begin{align*}
        \int_t^{t + \frac{2 \pi}{h}} S^\top \Ad_{X_d(\tau)}^{\vee} B  & B^\top \Ad_{X_d(\tau)}^{\vee \top} S \td \tau = \begin{pmatrix}
                \frac{8 \pi}{h} & 0 & 0\\
                0 & \frac{(b^2 + 1)\pi}{h}& 0\\
                0 & 0 & \frac{(a^2 + 1)\pi}{h}
        \end{pmatrix} \geq \varepsilon \Id,
\end{align*}
where $\varepsilon$ is any positive number satisfying $ \varepsilon \leq \frac{\pi}{h}\min\{8, a^2 + 1, b^2 + 1\}$.
Therefore, the ellipse is persistently exciting and can be stabilised with the proposed controller.

\subsection{Results}
The specific trajectory to be tracked is given by:
\begin{align*}
        p_d(t) = \begin{pmatrix}
        3 \cos(\frac{2\pi}{5}t) \\  5 \sin(\frac{2\pi}{5}t)
\end{pmatrix}.
\end{align*}

The corresponding inputs $v_d, \Omega_d$ can be recovered by
\begin{align*}
        v_d(t) &= \lVert \dot{p}_d \rVert = \sqrt{9 \cos^2 (\frac{2\pi}{5}t) + 25 \sin^2(\frac{2\pi}{5}t)},\\
        \Omega_d(t) &= \frac{2 \pi}{5}.
\end{align*}

The simulation is run in two configurations: firstly, with the origin of the ellipse at $(0, 0)$, and then with the origin at $(3, 3)$.
In both cases, the initial system state is perturbed to $p(0) = p_d(0) + ( 3.0, - 2.0), \theta(0) = \theta_d(0) + \frac{\pi}{2}$, so that initial relative error is the same.
The simulation is repeated for the controllers proposed in \cite{kanayamaStableTrackingControl1990} and \cite{rodriguez-cortesNewGeometricTrajectory2022} in order to obtain comparative results.
The simulation results are shown in Figure \ref{fig:centered} for the choice of a $(0, 0)$ origin and in Figure \ref{fig:offset} for the choice of a $(3, 3)$ origin.

\begin{figure}[!tb]
        \includegraphics[width=0.7\linewidth]{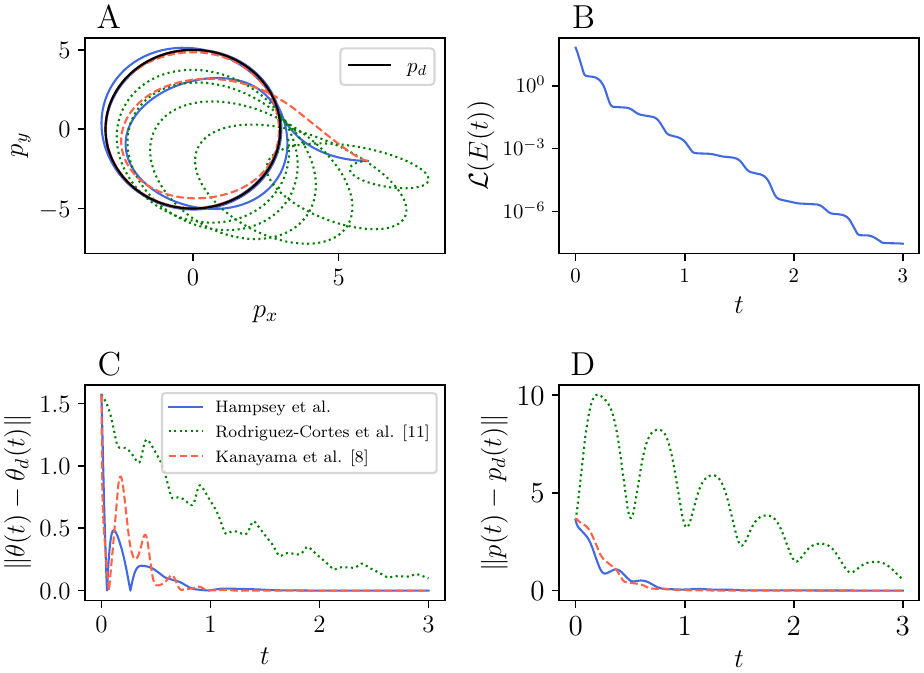}
        \centering
        \caption{Results with ellipse origin at (0, 0). (a) Desired (black) and actual position trajectories for the proposed controller and controllers from literature. (b) Lyapunov function vs time for the proposed controller. (c) Heading error vs time for the proposed controller and other controllers from literature. (d) Position error vs time for the proposed controller and other controllers from literature.}
        \label{fig:centered}
\end{figure}

\begin{figure}[!tb]
        \includegraphics[width=0.7\linewidth]{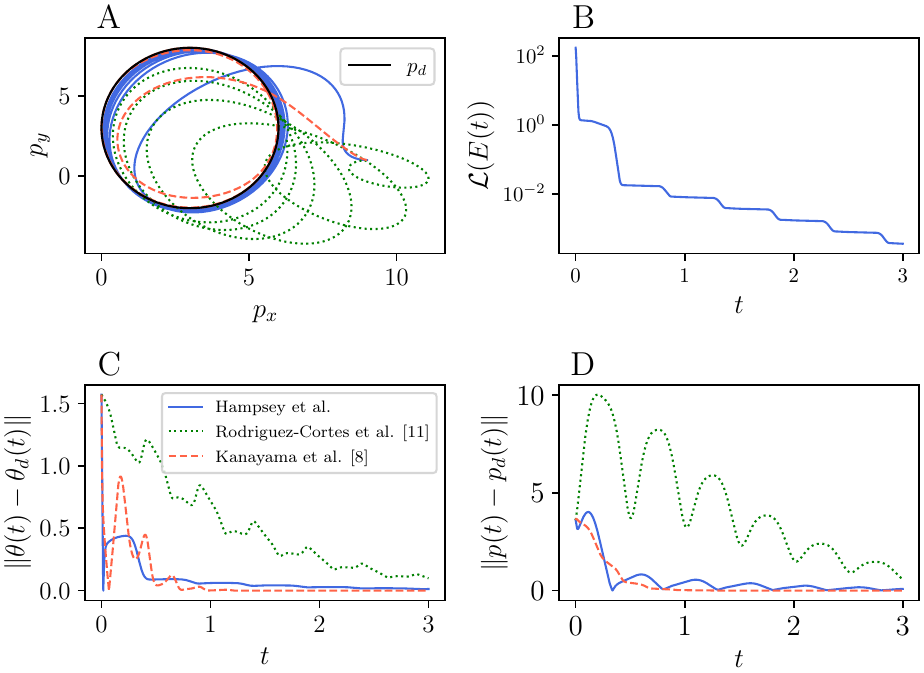}
        \centering
        \caption{Results with ellipse origin at (3, 3). (a) Desired (black) and actual position trajectories for the proposed controller and controllers from literature. (b) Lyapunov function vs time for the proposed controller. (c) Heading error vs time for the proposed controller and other controllers from literature. (d) Position error vs time for the proposed controller and other controllers from literature.}
        \label{fig:offset}
\end{figure}

\subsection{Discussion}
The simple simulations provided (Figure~\ref{fig:centered}, Figure~\ref{fig:offset}) verify the proposed control design.
In both cases, the Lyapunov function is globally non-increasing and is linear in logarithmic coordinates, showing the local exponential stability.
When the ellipse is centered on $(0, 0)$ (Figure~\ref{fig:centered}), we note that the performance of the proposed controller is similar to other controllers reported in the literature (\cite{kanayamaStableTrackingControl1990}, \cite{rodriguez-cortesNewGeometricTrajectory2022}), although this is dependent on gain tuning.
It is interesting to note that although the spatial error may appear unconventional (Figure~\ref{fig:rightInvariantError}), the convergence of the trajectory appears quite natural when the ellipse is centered on $(0, 0)$.
Comparing Figure~\ref{fig:centered} to the results with an ellipse centered on $(3, 3)$ (Figure~\ref{fig:offset}), the effect of the dependence of the spatial error on the inertial frame becomes evident.
Here, our proposed controller takes a longer time to converge in both heading and position, and follows a noticeably different trajectory (although the performance is still comparable with the other used controllers).
In contrast, the controllers from literature, which use a body-fixed error, follow the exact same trajectory.
For our proposed controller, the fact that the trajectory depends on the choice of origin with respect to which the error is computed opens an interesting research direction in how to choose a good origin for a particular trajectory tracking problem.

\section{Conclusion}

The choice of error for systems on Lie groups is not merely a matter of taste, but results in different error systems.
The properties of the emergent error dynamics are closely tied to the physical meaning of the error choice.
We have provided physical interpretations of the body-fixed and spatial, or left and right-invariant errors for a rigid body with an $\SE(2)$ symmetry.
We have shown that, for the kinematic unicycle, the choice of the spatial error leads to a simple gradient-based controller design.
We have proved the almost-global asymptotic convergence of this control scheme for a class of persistently exciting trajectories.
Finally, we have verified the effectiveness of this control in simulation.

\addtolength{\textheight}{-12cm}   



\section*{APPENDIX}

\begin{lemma}
        Let $A : \R_{\geq 0} \to \R^{n \times m}$ and $C: \R_{\geq 0} \to \R^{m}$ be bounded functions. If $A(t)C(t) \to 0$ and $\dot{C}(t) \to 0$ as $t \to \infty$, and $A(t)$ is persistently exciting, then $C(t) \to 0$.
        \label{lemma:persistent_excitation}
\end{lemma}
\begin{proof}
        Let $T > 0$ be such that the persistent excitation condition holds.
        We first show that $\lim_{t \to \infty}  C(t)^\top \int_t^{t + T} A(\tau)^\top A(\tau)C(\tau) \td \tau$ = 0.
        Let $L = \sup_t \lVert C(t) \rVert$ and $M = \sup_t \lVert  A(t) \rVert$.
        Given $\varepsilon > 0$, let $T' > 0$ be large enough such that $\lVert A(t)C(t) \rVert < \frac{\varepsilon}{LMT}$ for $t > T'$.
        Then for $t > T'$, it holds that
        \begin{align*}
                \lVert C(t)^\top  \int_t^{t + T} &A(\tau)^\top A(\tau)C(\tau) \td \tau \rVert \leq \lVert C(\tau)^\top \rVert  \int_t^{t + T} \lVert A(\tau)^\top\rVert  \lVert A(\tau)C(\tau) \rVert \td \tau < L \int_t^{t + T} M \frac{\varepsilon}{LMT} = \varepsilon,
        \end{align*}
        so $C(t)^\top \int_t^{t + T} A(\tau)^\top A(\tau)C(\tau) \td \tau \to 0$.
        On the other hand, by noting that for $\tau \in [t, t+T]$, $C(\tau) = C(t) + \int_t^\tau \dot{C}(s) \td s$, it also holds that
        \begin{align*}
                C(t)^\top \int_t^{t + T} &A(\tau)^\top A(\tau)C(\tau) \td \tau =C(t)^\top \left( \int_t^{t + T} A(\tau)^\top A(\tau) \td \tau \right) C(t)  +C(t)^\top \int_t^{t + T} A(\tau)^\top A(\tau)\int_t^\tau \dot{C}(s) \td s \td \tau.
        \end{align*}
        The term $C(t)^\top \int_t^{t + T} A(\tau)^\top A(\tau)\int_t^\tau \dot{C}(s) \td s \td \tau$ can also be shown to approach 0 as $t \to \infty$:
        let $L$ and $M$ be defined as before and given $\varepsilon > 0$, let $T' > 0$ be large enough such that $\lVert \dot{C}(t) \rVert < \frac{2\varepsilon}{LM^2T^2} $ for $t > T'$.
        Then for $t > T'$, it holds that
        \begin{align*}
                \lVert C(t)^\top &\int_t^{t + T} A(\tau)^\top A(\tau)\int_t^\tau \dot{C}(s) \td s \td \tau  \rVert < L \int_t^{t + T} M^2 \int_t^\tau \frac{2\varepsilon}{LM^2T^2} \td s \td \tau = LM^2 \frac{2\varepsilon}{LM^2T^2} \int_t^{t + T} (\tau - t) \dtau = \varepsilon,
        \end{align*}
        so $\int_t^{t + T} A(\tau)^\top A(\tau)\int_t^\tau \dot{C}(s) \td s \td \tau \to 0$.
        Thus, the remaining term of the integral, $C(t)^\top \left( \int_t^{t + T} A(\tau)^\top A(\tau) \td \tau \right) C(t)$, must also approach zero as $t \to \infty$.
        But because $A(t)$ is exciting, the integral $\int_t^{t + T} A(\tau)^\top A(\tau) \td \tau$ is lower bounded by $\alpha \Id$, for some $\alpha > 0$, for all $t$.
        So
        \begin{align*}
                0 \leq \alpha \lVert C(t) \rVert^2 \leq C(t)^\top \left( \int_t^{t + T} A(\tau)^\top A(\tau) \td \tau \right) C(t)
        \end{align*} for all $t$ and hence $C(t) \to 0$.
    \end{proof}

\begin{proposition}
        Consider the linear time-varying system
        \begin{align}
        \dot{x} = -A(t) x, \label{eq:lin_sys_2}
        \end{align}
        with $A(t) \in \R^{n \times n}$ symmetric and positive semi-definite.
        Assume that there exist a $T > 0$ and $\varepsilon > 0$ such that for all $t$,
        \begin{align}
        \int_t^{t+T} A(\tau) \td \tau \geq \varepsilon \Id. \label{eq:excited_linear}
        \end{align}
        Then \eqref{eq:lin_sys_2} is uniformly exponentially stable at $x = 0$.
        \label{prop:excited_linear}
\end{proposition}
\begin{proof}
        We aim to show that the assumed conditions of symmetric positive semi-definite $A(t)$ and \eqref{eq:excited_linear} imply the conditions required of \cite{morganUniformAsymptoticStability1977} (Theorem 1).
        This proof reproduces an argument used in \cite{trumpfAnalysisNonLinearAttitude2012} (Proposition 4.6).
        By the Cauchy-Schwarz inequality, the fact that $\int_t^{t+T} A(\tau) \td \tau$ is lower-bounded implies that for any unit vector $y \in \R^n$, the integral $\int_t^{t+T} \lVert A(\tau) y \rVert \td \tau$ is also lower bounded, as
        \begin{align*}
            \int_t^{t+T} \lVert A(\tau) y \rVert \td \tau &= \int_t^{t+T} \lVert y \rVert \lVert A(\tau) y \rVert \td \tau \geq \int_t^{t+T} y^\top A(\tau) y \td \tau = y^\top \int_t^{t+T}  A(\tau) \td \tau y \geq y^\top (\varepsilon \Id) y = \varepsilon.
        \end{align*}

        Given $t \geq t_0 > 0$, define $N = \lfloor \frac{t - t_0}{T} \rfloor $ and consider subdividing the interval $[t_0, t]$ into $N$ intervals of length $T$ and a remainder $[t_0 + NT, t]$:
        \begin{align*}
            [t_0, t] &= [t_0, t_0 + T] \cup [t_0 + T, t_0 + 2T] \cup... [t_0 + (N-1)T, t_0 + NT] \cup [t_0 + NT, t].
        \end{align*}
        Then
        \begin{align*}
            \int_{t_0}^t  \lVert A(\tau) y \rVert \td \tau &= \sum_{i = 1}^{N} \int_{t_0 + (i-1)T}^{t_0 + iT} \lVert A(\tau) y \rVert \td \tau + \int_{t_0 + NT}^{t} \lVert A(\tau) y \rVert \td \tau\\
                                          &\geq N\varepsilon\\
                                          &=  \frac{\varepsilon}{T} \left\lfloor \frac{t - t_0}{T} \right\rfloor T\\
                                          &= \frac{\varepsilon}{T} \left( (t - t_0) -((t - t_0) - \left\lfloor \frac{t - t_0}{T} \right\rfloor T) \right)\\
                                          &\geq \frac{\varepsilon}{T} \left( (t - t_0) -T \right)\\
                                          &= \frac{\varepsilon}{T} (t - t_0) - \varepsilon.
        \end{align*}
        This is the condition required by \cite{morganUniformAsymptoticStability1977} (Theorem 1) and so the system \eqref{eq:excited_linear} is uniformly asymptotically stable.
        Then by \cite{haleOrdinaryDifferentialEquations2009} (Theorem III.2.1), the system is also uniformly exponentially stable.
    \end{proof}



\bibliographystyle{plainnat}
\bibliography{refs}

\end{document}

%% file: preamble_compatible.tex
\usepackage{xcolor}  

\usepackage{amssymb}

\usepackage{amsmath}
\usepackage{mathdots}

\usepackage{mathtools}

\usepackage{tensor}

\usepackage{urwchancal}
\DeclareFontFamily{OT1}{pzc}{}
\DeclareFontShape{OT1}{pzc}{m}{it}{<-> s * [1.10] pzcmi7t}{}
\DeclareMathAlphabet{\mathpzc}{OT1}{pzc}{m}{it}

\usepackage{graphicx}
\usepackage{ifpdf}
\ifpdf
\usepackage{epstopdf}
\PrependGraphicsExtensions{.svg}
\fi

\newtheorem{theorem}{Theorem}[section]
\newtheorem{lemma}[theorem]{Lemma}

\newtheorem{proposition}[theorem]{Proposition}


%% file: notation_defs.tex
\providecommand{\R}{\mathbb{R}}

\providecommand{\SO}{\mathbf{SO}}

\providecommand{\SE}{\mathbf{SE}}

\providecommand{\grpG}{\mathbf{G}}

\providecommand{\so}{\mathfrak{so}}
\providecommand{\se}{\mathfrak{se}}

\providecommand{\Sph}{\mathrm{S}}

\providecommand{\calE}{\mathcal{E}}

\providecommand{\calL}{\mathcal{L}}

\providecommand{\Id}{I} 

\providecommand{\eb}{\mathbf{e}} 

\DeclareMathOperator{\tr}{tr}

\DeclareMathOperator{\diag}{diag}

\DeclareMathOperator{\Ad}{Ad}

\providecommand{\td}{\mathrm{d}}

\providecommand{\dtau}{\td \tau}

\usepackage{accents}
\makeatletter
\providecommand{\scirc}{%
    \hbox{\fontfamily{\rmdefault}\fontsize{0.4\dimexpr(\f@size pt)}{0}\selectfont{\raisebox{-0.52ex}[0ex][-0.52ex]{$\circ$}}}}

\makeatother

\makeatletter
\providecommand{\rcirc}{%
	\hbox{\fontfamily{\rmdefault}\fontsize{0.4\dimexpr(\f@size pt)}{0}\selectfont{\raisebox{-0.2ex}[0ex][-0.52ex]{$\circ$}}}}

\makeatother

\mathchardef\mhyphen="2D

